%% file: main.tex
\documentclass[letterpaper,10pt,conference]{ieeeconf}

\usepackage{cite}
\usepackage{graphicx}
\usepackage{tikz}
\usepackage{subcaption}
\usepackage[cmex10]{amsmath}
\usepackage{algpseudocode}
\usepackage{algorithmicx} 
\usepackage{array}
\usepackage{url}
\usepackage{mathrsfs}

\usepackage{epsfig}
\usepackage{amsfonts,amssymb,multirow,bigstrut,booktabs,ctable,latexsym}

\usepackage[mathscr]{eucal}
\usepackage{verbatim}

\usepackage{amsthm}
\usepackage{algorithm}

\usepackage{stmaryrd}

\IEEEoverridecommandlockouts  
\begin{document}

\newtheorem{definition}{Definition}
\newtheorem{lemma}{Lemma}
\newtheorem{corollary}{Corollary}
\newtheorem{theorem}{Theorem}
\newtheorem{example}{Example}
\newtheorem{proposition}{Proposition}
\newtheorem{remark}{Remark}
\newtheorem{assumption}{Assumption}
\newtheorem{corrolary}{Corrolary}
\newtheorem{property}{Property}
\newtheorem{ex}{EX}
\newtheorem{problem}{Problem}
\newcommand{\argmin}{\arg\!\min}
\newcommand{\argmax}{\arg\!\max}
\newcommand{\st}{\text{s.t.}}
\newcommand \dd[1]  { \,\textrm d{#1}  }

\title{\Large\bf Safety-Critical Control Synthesis for Unknown Sampled-Data Systems via Control Barrier Functions}

\author{Luyao Niu, Hongchao Zhang, and Andrew Clark %
\thanks{L. Niu, H. Zhang, and A. Clark are with the Department of Electrical and Computer Engineering, Worcester Polytechnic Institute, Worcester, MA 01609 USA.
{\tt\small \{lniu,hzhang9,aclark\}@wpi.edu}}
	\thanks{This work was supported by the National Science Foundation and the Office of Naval Research via grants CNS-1941670 and N00014-17-1-2946.}
}
\thispagestyle{empty}
\pagestyle{empty}

\maketitle

\begin{abstract}
As the complexity of control systems increases, safety becomes an increasingly important property since safety violations can damage the plant and put the system operator in danger. When the system dynamics are unknown, safety-critical synthesis becomes more challenging. Additionally, modern systems are controlled digitally and hence behave as sampled-data systems, i.e., the system dynamics evolve continuously while the control input is applied at discrete time steps. In this paper, we study the problem of control synthesis for safety-critical sampled-data systems with unknown dynamics. We overcome the challenges introduced by sampled-data implementation and unknown dynamics by constructing a set of control barrier function (CBF)-based constraints. By satisfying the constructed CBF constraint at each sampling time, we guarantee the unknown sampled-data system is safe for all time. We formulate a non-convex program to solve for the control signal at each sampling time. We decompose the non-convex program into two convex sub-problems. We illustrate the proposed approach using a numerical case study.
\end{abstract}

\input{intro}
\input{related.tex}

\input{prelim.tex}
\input{formulation.tex}
\input{sol}
\input{simulation}

\input{Conclusion}

\bibliographystyle{IEEEtran}
\bibliography{IEEEabrv,MyBib}

\input{Appendix}

\end{document}

%% file: intro.tex
\section{Introduction}

Safety-critical cyber-physical systems (CPSs) are found in applications such as autonomous vehicles and advanced manufacturing. The safety property is typically formulated as forward invariance of a given safe set. Safety violations could lead to severe damage to the controlled plant or danger to human operators \cite{knight2002safety}. Control synthesis for safety-critical CPSs has been extensively studied in existing literature when the models of the CPSs are known \cite{ames2016control,wang2017safety,cohen2020approximate}. 

Several factors may cause safety violations of the plant even when the nominal controller is designed to be safe. One challenge is raised by the digital/discrete implementation of a continuous-time system with continuous input. In a practical digital implementation, the system state is only observable at each sampling time, and the control signal is applied in a zero-order hold (ZOH) manner during each sampling period. That is, the system is implemented as a sampled-data system. Another challenge is that the system models used to design controllers are not perfect in practice, and thus there exist uncertainties in the system model. Due to these uncertainties, the designed controller may fail to guarantee safety even if safety is guaranteed for the nominal  plant.

These two challenges have been studied separately. Safety-critical control synthesis has been studied for sampled-data system with known dynamics \cite{singletary2020control,cortez2019control,breeden2021control}. For unmodeled systems with continuous-time or discrete-time dynamics, safety-critical control synthesis has been studied by assuming the existence of a backup controller that ensures safety of the unknown system \cite{mannucci2017safe,folkestad2020data, taylor2020learning}, or a well-calibrated model, e.g., a model learned using a Gaussian process \cite{jagtap2020control,wang2018safe,berkenkamp2017safe}. Although high-probability safety guarantees can be achieved by the methods described in \cite{jagtap2020control,wang2018safe,berkenkamp2017safe}, they suffer from a tradeoff between overly conservative learned model when faced with large uncertainty and the potential safety violation when failing to capture the true dynamics. To the best of our knowledge, jointly addressing these two challenges without prior knowledge of a safe backup controller or a well-calibrated model has not been studied.

In this paper, we study safety-critical control synthesis for a sampled-data system with unknown dynamics. We address the challenges introduced due to the sampled-data system and unknown dynamics by developing a set control barrier function (CBF) constraints at each sampling time. A CBF constraint is an inequality that is imposed on the control signal, whose satisfaction implies forward invariance. We estimate the constructed CBF constraints by bounding the reachable set during each sampling period and calculating an interval that contains the system dynamics, leveraging the Lipschitz continuity assumption on the system dynamics. By satisfying these CBF constraints at each sampling time, the system is guaranteed to be safe for all time. To summarize, this paper makes the following contributions.  
\begin{itemize}
    \item We construct a CBF constraint for the unknown sampled-data system. We provide a sufficient condition for satisfying the CBF constraint by bounding the set of reachable states for each sampling period and calculating an interval which contains the unknown system dynamics.
    \item We formulate a non-convex optimization problem subject to the constructed CBF constraints to calculate the ZOH control signal for each sampling period. We solve the non-convex optimization problem by proposing a two-stage approach which only involves convex programs.
    \item We prove that the synthesized controller ensures the system is safe with respect to the given safe set.
    \item We validate our proposed framework using a numerical case study on a DC motor. We show that our proposed approach ensures the safety of the DC motor.
\end{itemize}

The remainder of this paper is organized as follows. We review the related work and preliminary background in Section \ref{sec:related} and Section \ref{sec:preliminary}, respectively. The system model and problem formulation are presented in Section \ref{sec:formulation}. We give the solution approach in Section \ref{sec:sol}, and illustrate the proposed approach using a numerical case study in Section \ref{sec:simulation}. Section \ref{sec:conclusion} concludes this paper.

%% file: related.tex
\section{Related Work}\label{sec:related}
Multiple approaches have been proposed for safety-critical control synthesis for CPSs with known dynamics, including Hamilton-Jacobi-Bellman-Isaacs (HJI) equation \cite{tomlin1998conflict}, mixed-integer program \cite{mellinger2012mixed}, and control barrier function (CBF) and control Lyapunov function (CLF) -based methodologies \cite{ames2016control,wang2017safety,cohen2020approximate}. The CBF-based methods formulate a quadratic program to calculate the controller, assuming the system state remains unchanged during the discrete time interval. 

For sampled-data systems, CBF-based approaches have shown great success \cite{singletary2020control,cortez2019control,breeden2021control}. The authors of \cite{singletary2020control} and \cite{breeden2021control} focus on sampled-data systems with known dynamics. In \cite{cortez2019control}, a sampled-data system with an additive disturbance is studied. In this work, we consider a sampled-data system with unknown dynamics. CBF-based methods normally require the knowledge of the system dynamics to calculate the CBF constraint. The unknown dynamics lead to a scenario where calculating the CBF constraints is not feasible, and thus makes the approaches proposed in \cite{singletary2020control,cortez2019control,breeden2021control} not applicable.

Learning-based control algorithms have been proposed to address systems that contain unknown uncertainties. Recent works have demonstrated the success of CBF- and/or CLF- based method along with learning algorithms \cite{jagtap2020control,cheng2019end,choi2020reinforcement,taylor2020learning}. This category of approaches leverages the forward invariance and stability properties provided by barrier functions and Lyapunov functions, respectively. However, they assume that there exists a well-calibrated model of the unknown system \cite{berkenkamp2017safe} and a safe backup controller to recover from failure \cite{mannucci2017safe}. 
Moreover, CBF-based learning approaches have to handle the tradeoff between the overly constrained learned model and failure to capture the true dynamics.

Reachable set learning aims at learning the set of reachable states of the system so as to compute a controller that gives no intersection between the reachable states and the unsafe region \cite{fisac2018general}. The computation of reachable sets relies on numerically solving HJI equations, which incurs high computational complexity and poor scalability \cite{Gillula2011guaranteed,Gillula2012guaranteed}. A Gaussian Process based reachability analysis is proposed in \cite{Akametalu2014reachability} to compute the reachable set. Compared with forward reachable set computation, region-of-attraction focuses on computing the set of states starting from which the system is guaranteed to be safe \cite{berkenkamp2017safe}. These learning-based approaches focus on either continuous-time or discrete-time systems. In this work, we study safety-critical control synthesis for the unknown sampled-data system.


%% file: prelim.tex
\section{Preliminary Background}\label{sec:preliminary}
\subsection{Control Barrier Function}
A continuous function $\alpha:[0,a)\mapsto[0,\infty)$ belongs to class $\mathcal{K}$ if it is strictly increasing and $\alpha(0)=0$. A continuous function $\alpha:[-b,a)\mapsto(-\infty,\infty)$ is said to belong to extended class $\mathcal{K}$ if it is strictly increasing and $\alpha(0)=0$ for some $a,b>0$.

Consider a continuous-time control-affine system 
\begin{equation}\label{eq:dynamic}
    \dot{x}_t=f(x_t)+g(x_t)u_t
\end{equation}
where $x_t\in\mathcal{X}\subseteq\mathbb{R}^n$ is the system state and $u_t\in\mathcal{U}\subseteq\mathbb{R}^m$ is input provided by the controller. Vector-valued and matrix-valued functions $f(x_t)$ and $g(x_t)$ are of appropriate dimensions. Let a safe set $\mathcal{C}$ be defined as
\begin{equation}\label{eq:safe set}
    \mathcal{C}=\{x\in\mathcal{X}:h(x)\geq0\},
\end{equation}
where $h:\mathcal{X}\mapsto\mathbb{R}$ is a continuously differentiable function. We say system \eqref{eq:dynamic} is safe with respect to $\mathcal{C}$ if $x_t\in\mathcal{C}$ for all time $t\geq 0$.

CBF-based approaches have been used to guarantee safety of system \eqref{eq:dynamic} with respect to safe set $\mathcal{C}$. We give the definition of zeroing CBF as follows.
\begin{definition}[Zeroing CBF (ZCBF) \cite{ames2016control}]\label{def:ZCBF}
Consider a dynamical system \eqref{eq:dynamic} and a continuously differentiable function $h:\mathcal{X}\mapsto\mathbb{R}$. If there exists a locally Lipschitz extended class $\mathcal{K}$ function $\alpha$ such that for all $x\in\mathcal{X}$ the following inequality holds
\begin{equation*}
    \sup_{u\in\mathcal{U}}\bigg\{\frac{\partial h(x)}{\partial x}f(x)+\frac{\partial h(x)}{\partial x}g(x)u
    +\alpha(h(x))\bigg\}\geq 0,
\end{equation*}
then function $h$ is a ZCBF.
\end{definition}
In this paper, we focus on ZCBF. Extensions to higher relative degree CBFs are subject to our future work. A sufficient condition for the safety guarantee can be derived using ZCBF as follows.
\begin{lemma}[\cite{ames2016control}]\label{thm:ZCBF}
Given a dynamical system \eqref{eq:dynamic} and a safe set \eqref{eq:safe set} defined by some continuously differentiable function $h:\mathcal{X}\mapsto\mathbb{R}$, if $h$ is a ZCBF defined on $\mathcal{X}$, then $\mathcal{C}$ is forward invariant.
\end{lemma}
Using Lemma \ref{thm:ZCBF}, one can solve for the controller at each time using a quadratic program \cite{ames2016control}
\begin{subequations}\label{eq:QP}
\begin{align}
    \min_u~&u^\top R(x) u + Q(x)^\top u\\
    \st~&\frac{\partial h(x)}{\partial x}f(x)+\frac{\partial h(x)}{\partial x}g(x)u
    +\alpha(h(x))\geq 0\label{eq:ZCBF def}\\
    &u\in\mathcal{U}
\end{align}
\end{subequations}
where $R(x)\in\mathbb{R}^m$ is positive definite and $Q(x)\in\mathbb{R}^{m}$. 

\subsection{Notations}
Let $x$ be a vector and $f(x)$ be a vector-valued function, we use $x_j$ and $f_j(x)$ to denote their $j$-th component, respectively. Let $A$ be a matrix. We use $A_{i,j}$ to denote its element at $i$-th row and $j$-th column. Let $x_t$ be a vector at time $t$. We use $x_{t,j}$ to denote the $j$-th component of $x_t$. Comparisons between vectors are implemented element-wise. Bold symbols are used to represent intervals.

%% file: formulation.tex
\section{Problem Formulation}\label{sec:formulation}

Consider a continuous-time control-affine system in the form of \eqref{eq:dynamic}. The system contains uncertainties and hence $f(x_t)$ and $g(x_t)$ are unknown. We define a feedback controller $\mu:\mathcal{X}\mapsto\mathcal{U}$ to be a function that maps the system state to a control input. Given the current system state $x_t$ at time $t$ and a feedback controller $\mu$, we denote the system state at time $t'$ as $\varphi^{t'}(x_t,\mu)$. The system is given a safe set as defined in \eqref{eq:safe set}.

We consider the sampled-data implementation of system \eqref{eq:dynamic}. That is, the system is sampled using a sampling period $\Delta t$. Only the system states $x_{z\Delta t}$ at the sampling time are known, where $z=0,1,\ldots$. At each sampling time, a zero-order hold (ZOH) feedback controller $\mu(x_{z\Delta t})$ is applied to system \eqref{eq:dynamic}. In other words, $u_t=\mu(x_{z\Delta t})$ for all $t\in[z\Delta t,(z+1)\Delta t)$.

We have a data set of system \eqref{eq:dynamic} as side information. Let $K\in\mathbb{N}$ and $K\geq 1$. We denote a finite set of $K$ samples of state-input pairs as $R_K=\{(x_{t_k},u_{t_k},x_{t_{k+1}})\}_{k=1}^K$ where $x_{t_{k+1}}=\varphi^{t_{k+1}}(x_{t_k},\mu)$. Here $\mu$ represents a zero-order hold (ZOH) input $u_t=u_{t_k}$ for all time $t\in[t_k,t_{k+1})$. We assume that $h(x_{t_k})\geq 0$ for all $k=1,\ldots,K$. 

In the following, we formally state our assumptions.
\begin{assumption}\label{assump:Lipschitz}
We assume that functions $f_j(x)$ and $g_{j,s}(x)$ are Lipschitz continuous with Lipschitz constants $L_{f_j}$ and $L_{g_{j,s}}$, respectively, for all $j=1,\ldots,n$ and $s=1,\ldots,m$. The Lipschitz constants are known. We assume that $\sup_{x\in\mathcal{C}}\|f(x)\|$ and $\sup_{x\in\mathcal{C}}\|g(x)\|$ are given.
\end{assumption}
Lipschitz continuity is a commonly made assumption for reachability and safety analysis \cite{taylor2020learning,wang2017safety,khalil2002nonlinear}. The assumption on the bound of system dynamics is often seen in the worst-case safety analysis \cite{fisac2018general,berkenkamp2017safe}.
\begin{assumption}\label{assump:compact}
We assume that the safe set $\mathcal{C}$ and control input set $\mathcal{U}$ are compact. Additionally, control input set $\mathcal{U}$ is convex.
\end{assumption}
The problem studied in this work is as follows.

\begin{problem}\label{prob:worst-case safety}
Given a finite set of samples $R_K=\{(x_{t_k},u_{t_k})\}_{k=1}^K$ for some $K\in\mathbb{N}$ generated by implementing a given control input $u_{t_k}$ for all time $t\in[t_k,t_{k+1})$ to system \eqref{eq:dynamic} whose dynamics are unknown, synthesize a ZOH feedback controller $\mu$ such that system \eqref{eq:dynamic} is safe with respect to set $\mathcal{C}=\{x:h(x)\geq 0\}$.
\end{problem}

%% file: sol.tex
\section{Solution Approach}\label{sec:sol}

Our solution approach leverages Lemma \ref{thm:ZCBF} to guarantee safety of the system. We first construct a CBF constraint for the unknown sampled-data system to ensure the safety. Then we calculate a bound for the unknown system dynamics to evaluate the constructed CBF constraint. Finally, we formulate an optimization problem to solve for the control signal at each sample time.

\subsection{Construction of CBF Constraints for Unknown Systems}
When the system model is known and the system state is observable for all time $t\geq 0$, safety-critical synthesis can be achieved efficiently using quadratic program \eqref{eq:QP}. We consider the sampled-data system with unknown dynamics, which makes it difficult to evaluate the constraint given in \eqref{eq:ZCBF def}. In this subsection, we construct a CBF constraint that can be evaluated at each sampling time for unknown sampled-data system to guarantee that \eqref{eq:ZCBF def} holds for all time $t\in[z\Delta t,(z+1)\Delta t)$ for each sampling period $z=0,1,\ldots$, and hence guarantee system safety. 

Inspired by \cite{cortez2019control}, for any $t\in[z\Delta t,(z+1)\Delta t)$, we define
\begin{multline}\label{eq:CBF error}
    e(x_t,x_{z\Delta t},u_{z\Delta t})=\frac{\partial h(x_{z\Delta t})}{\partial x}[f(x_{z\Delta t})+g(x_{z\Delta t})u_{z\Delta t}]
    \\+\alpha(h(x_{z\Delta t}))
    -\frac{\partial h(x_t)}{\partial x}[f(x_t)-g(x_t)u_{z\Delta t}]-\alpha(h(x_t)).
\end{multline}
The definition given in \eqref{eq:CBF error} models the difference between the CBF constraints evaluated at states $x_{z\Delta t}$ and $x_t$ when control input $u_{z\Delta t}$ is applied. Given \eqref{eq:CBF error}, we have that
\begin{subequations}\label{eq:CBF with margin}
\begin{align}
    &\frac{\partial h(x_t)}{\partial x}f(x_t)+\frac{\partial h(x_t)}{\partial x}g(x_t)u_t+\alpha(h(x_t))\nonumber\\
    =&\frac{\partial h(x_{z\Delta t})}{\partial x}f(x_{z\Delta t})+\frac{\partial h(x_{z\Delta t})}{\partial x}g(x_{z\Delta t})u_{z\Delta t}
    \nonumber\\
    &\quad\quad\quad\quad\quad\quad+\alpha(h(x_{z\Delta t}))-e(x_t,x_{z\Delta t},u_{z\Delta t})\\
    \geq &\frac{\partial h(x_{z\Delta t})}{\partial x}f(x_{z\Delta t})+\frac{\partial h(x_{z\Delta t})}{\partial x}g(x_{z\Delta t})u_{z\Delta t}
    \nonumber\\
    &\quad+\alpha(h(x_{z\Delta t}))-\max_{x_t,x_{z\Delta t},u_{z\Delta t}}|e(x_t,x_{z\Delta t},u_{z\Delta t})|
\end{align}
\end{subequations}
If we can guarantee that the right-hand side of \eqref{eq:CBF with margin} is non-negative, then safety of system \eqref{eq:dynamic} holds by Lemma \ref{thm:ZCBF}. We define the following quantities:
\begin{equation}\label{eq:theta}
    \theta(u) = \sqrt{\sum_{j=1}^n\left(L_{f_j}+\sum_{s=1}^mL_{g_{j,s}}|u_s|\right)^2},\quad\Theta=\max_{u\in\mathcal{U}}\theta(u).
\end{equation}
The existence of $\Theta$ is guaranteed by Assumption \ref{assump:compact}. In the following, we bound $\max_{x_t,x_{z\Delta t},u_{z\Delta t}}|e(x_t,x_{z\Delta t},u_{z\Delta t})|$ from above to calculate a lower bound for \eqref{eq:CBF with margin}.
\begin{lemma}\label{lemma:error bound}
Let $L_\alpha$ and $L_h$ be the Lipschitz constants of functions $\alpha$ and $h$, respectively. Let $\Theta$ be defined as in \eqref{eq:theta}. Then for any given $x_{z\Delta t}$, $x_t$, and $u_{z\Delta t}$, we have
\begin{multline}
    |e(x_t,x_{z\Delta t},u_{z\Delta t})|\leq (L_h\Theta+L_\alpha)\|x_{z\Delta t}-x_t\|_2\\    
    +2L_h\|f(x_{z\Delta t})+g(x_{z\Delta t})u_{z\Delta t}\|_2.
\end{multline}
\end{lemma}
\begin{proof}
We bound $|e(x_t,x_{z\Delta t},u_{z\Delta t})|$ via
\begin{subequations}\label{eq:e bound 1}
\begin{align}
    &|e(x_t,x_{z\Delta t},u_{z\Delta t})|\nonumber\\
    =&\Big|\frac{\partial h(x_{z\Delta t})}{\partial x}[f(x_{z\Delta t})+g(x_{z\Delta t})u_{z\Delta t}]+\alpha(h(x_{z\Delta t}))\nonumber\\
    &-\frac{\partial h(x_t)}{\partial x}[f(x_t)-g(x_t)u_{z\Delta t}]-\alpha(h(x_t))\Big|\label{eq:e bound 1-1}\\
    =&\Big|\frac{\partial h(x_{z\Delta t})}{\partial x}[f(x_{z\Delta t})+g(x_{z\Delta t})u_{z\Delta t}]\nonumber\\
    &\quad-\frac{\partial h(x_t)}{\partial x}[f(x_{z\Delta t})+g(x_{z\Delta t})u_{z\Delta t}]+\frac{\partial h(x_t)}{\partial x}[f(x_{z\Delta t})\nonumber\\
    &\quad+g(x_{z\Delta t})u_{z\Delta t}]-\frac{\partial h(x_t)}{\partial x}[f(x_t)+g(x_t)u_{z\Delta t}]\nonumber\\
    &\quad\quad\quad\quad+\alpha(h(x_{z\Delta t}))-\alpha(h(x_t))\Big|\label{eq:e bound 1-2}\\
    \leq&\Big|\left(\frac{\partial h(x_{z\Delta t})}{\partial x}-\frac{\partial h(x_t)}{\partial x}\right)[f(x_{z\Delta t})+g(x_{z\Delta t})u_{z\Delta t}]\Big|\nonumber\\
    &+\Big|\frac{\partial h(x_t)}{\partial x}[f(x_{z\Delta t})+g(x_{z\Delta t})u_{z\Delta t}-f(x_t)-g(x_t)u_{z\Delta t}]\Big|\nonumber\\
    &\quad\quad\quad\quad\quad+|\alpha(h(x_{z\Delta t}))-\alpha(h(x_t))|\label{eq:e bound 1-3}
\end{align}
\end{subequations}
where \eqref{eq:e bound 1-1} holds by definition given in \eqref{eq:CBF error}, \eqref{eq:e bound 1-2} holds by adding and subtracting $\frac{\partial h(x_t)}{\partial x}[f(x_{z\Delta t})+g(x_{z\Delta t})u_{z\Delta t}]$, and \eqref{eq:e bound 1-3} holds by triangle inequality.

Since function $h$ is continuously differentiable, we have that $\|\frac{\partial h(x)}{\partial x}\|\leq L_h$, where $L_h$ is the Lipschitz constant of $h$. Hence we have that
\begin{multline}\label{eq:e bound 2}
    \Big|\left(\frac{\partial h(x_{z\Delta t})}{\partial x}-\frac{\partial h(x_t)}{\partial x}\right)[f(x_{z\Delta t})+g(x_{z\Delta t})u_{z\Delta t}]\Big|\\
    \leq 2L_h\|f(x_{z\Delta t})+g(x_{z\Delta t})u_{z\Delta t}\|_2.
\end{multline}

By the boundedness of $\frac{\partial h(x)}{\partial x}$ and Proposition \ref{prop:derivative bound} in the Appendix, we have that 
\begin{multline}\label{eq:e bound 3}
    \Big|\frac{\partial h(x_t)}{\partial x}[f(x_{z\Delta t})+g(x_{z\Delta t})u_{z\Delta t}-f(x_t)-g(x_t)u_{z\Delta t}]\Big|\\
    \leq L_h\Theta\|x_{z\Delta t}-x_t\|_2.
\end{multline}

Due to Lipschitz continuity of $\alpha(\cdot)$, we have that
\begin{equation}\label{eq:e bound 4}
    |\alpha(h(x_{z\Delta t}))-\alpha(h(x_t))|\leq L_\alpha\|x_{z\Delta t}-x_t\|_2.
\end{equation}

Substituting \eqref{eq:e bound 2} - \eqref{eq:e bound 4} into \eqref{eq:e bound 1-3} yields the lemma.
\end{proof}

Using Lemma \ref{lemma:error bound}, we can construct a CBF constraint as
\begin{multline}\label{eq:CBF with margin 1}
    \frac{\partial h(x_{z\Delta t})}{\partial x}f(x_{z\Delta t})+\frac{\partial h(x_{z\Delta t})}{\partial x}g(x_{z\Delta t})u_{z\Delta t}\\
    +\alpha(h(x_{z\Delta t}))-(L_h\Theta+L_\alpha)\|x_{z\Delta t}-x_t\|_2\\
    -2L_h\|f(x_{z\Delta t})+g(x_{z\Delta t})u_{z\Delta t}\|_2\geq 0.
\end{multline}
Using \eqref{eq:CBF with margin}, we have that if \eqref{eq:CBF with margin 1} holds, then Lemma \ref{thm:ZCBF} holds for all $t\in[z\Delta t,(z+1)\Delta t)$. However, since the system is unknown, we cannot compute $x_t$ and $\|f(x_{z\Delta t})+g(x_{z\Delta t})u_{z\Delta t}\|_2$, and thus we cannot calculate constraint \eqref{eq:CBF with margin 1}. In the subsequent subsection, we address this challenge.

\subsection{Sufficient Condition for Satisfying CBF Constraint \eqref{eq:CBF with margin 1}}

Since we consider sampled-data systems, only system states at time $z\Delta t$ with $z=0,1,\ldots$ are observable. Hence, $x_t$ in \eqref{eq:CBF with margin 1} is not known. Moreover, the term $\|f(x_{z\Delta t})+g(x_{z\Delta t})u_{z\Delta t}\|_2$ is not known since the system is unknown. In this subsection, we present how to estimate $x_t$ and $\|f(x_{z\Delta t})+g(x_{z\Delta t})u_{z\Delta t}\|_2$ to calculate the constructed CBF constraint given in \eqref{eq:CBF with margin 1}.

\subsubsection{Estimate System State $x_t$ During Sampling Period}

Although it is impractical to forward integrate the unknown dynamics to calculate $x_t$, we can bound $\|x_{z\Delta t}-x_t\|_2$. Using such bound and the observed system state $x_{z\Delta t}$, we can bound $x_t$ during each sampling period. We define
\begin{equation}\label{eq:beta}
    \beta=\sup_{x\in\mathcal{C},u\in\mathcal{U}} (\|f(x)+g(x)u\|).
\end{equation}
The existence of $\beta$ is guaranteed by Assumption \ref{assump:compact}. We can now bound $\|x_{z\Delta t}-x_t\|_2$ using the following proposition.
\begin{proposition}\label{prop:time bound}
Let $x_t\in\mathcal{X}$ and $\mu$ be a controller that specifies the control signal $u_t$ applied to system \eqref{eq:dynamic} for each time $t\in[T,T+\Delta t]$. We have that
\begin{equation}
    \left\|\varphi^t(x_t,\mu)-\varphi^T(x_t,\mu)\right\|_2\leq \frac{\|\beta\|_2}{\Theta}(e^{\Theta\Delta t}-1),
\end{equation}
where $\Theta$ is given in \eqref{eq:theta}, and $\beta$ is given in \eqref{eq:beta}.
\end{proposition}
\begin{proof}
By definition of $\varphi^t(x,\mu)$, we have that
\begin{subequations}\label{eq:time bound 1}
\begin{align}
    &\left\|\varphi^t(x_t,\mu)-\varphi^T(x_t,\mu)\right\|_2\nonumber\\
    =&\Big\|\int_T^t\left(f(x_\tau)+g(x_\tau)u_\tau\right)\dd\tau\Big\|_2\label{eq:time bound 1-1}\\
    \leq&\Big\|\int_T^t\left\{f(x_\tau)+g(x_\tau)u_\tau-f(x_T)-g(x_T)u_\tau\right\}\dd\tau\Big\|_2\nonumber\\
    &\quad\quad\quad\quad\quad+\left\|\int_T^t\left(f(x_T)+g(x_T)u_\tau\right)\dd\tau\right\|_2\label{eq:time bound 1-2}\\
    \leq&\int_T^t\theta(u_\tau)\|x_\tau-x_T\|_2\dd\tau+\|\beta\|_2(t-T).\label{eq:time bound 1-3}
\end{align}
\end{subequations}
where \eqref{eq:time bound 1-1} holds by definition of $\varphi^t(x_t,\mu)$, \eqref{eq:time bound 1-2} holds by triangle inequality, and \eqref{eq:time bound 1-3} holds by Proposition \ref{prop:derivative bound} in the Appendix and \eqref{eq:beta}.

Applying Gr\"{o}nwall's inequality \cite{bellman1943stability} to \eqref{eq:time bound 1} yields that
\begin{subequations}\label{eq:time bound 2}
\begin{align}
    &\left\|\varphi^t(x_t,\mu)-\varphi^T(x_t,\mu)\right\|_2\nonumber\\
    \leq&\int_T^t\theta(u_\tau)\|x_\tau-x_T\|_2\dd\tau+\|\beta\|_2(t-T)\label{eq:time bound 2-1}\\
    \leq&\|\beta\|_2(t-T)\nonumber\\
    &+\int_T^t\|\beta\|_2(t-T) \theta(u_\tau)\exp{\left(\int_\tau^t\theta(u_l)\dd l\right)} \dd\tau\label{eq:time bound 2-2}\\
    \leq &\|\beta\|_2(t-T)+\|\beta\|_2\Theta\int_T^t(\tau-T)e^{\Theta(t-\tau)} \dd\tau\label{eq:time bound 2-3}\\
    \leq& \frac{\|\beta\|_2}{\Theta}(e^{\Theta\Delta t}-1)\label{eq:time bound 2-4}
\end{align}
\end{subequations}
where \eqref{eq:time bound 2-1} holds by \eqref{eq:time bound 1}, \eqref{eq:time bound 2-2} holds by Gr\"{o}nwall's inequality \cite{bellman1943stability}, 
\eqref{eq:time bound 2-3} holds by \eqref{eq:theta} and calculating the inner integration, and \eqref{eq:time bound 2-4} holds by integration by parts and $t\in[T,T+\Delta t]$.
\end{proof}
Proposition \ref{prop:time bound} is closely related to \cite[Thm 3.4]{khalil2002nonlinear}. In \cite[Thm 3.4]{khalil2002nonlinear}, an upper bound of the distance between two nonlinear systems is established for each time, while Proposition \ref{prop:time bound} presents an upper bound of the distance between reachable states during a sampling period for the sampled-data system.

By Proposition \ref{prop:time bound}, we have that $\|x_{z\Delta t}-x_t\|_2\leq\frac{\|\beta\|_2}{\Theta}(e^{\Theta\Delta t}-1)$. For a sampled-data system, state $x_{z\Delta t}$ can be observed. Therefore, we can calculate a bound for $x_t$ for all $t\in[z\Delta t,(z+1)\Delta t)$ for all non-negative integer $z$.

\subsubsection{Estimate the Unknown System Dynamics}
In the following, we calculate a bound of $\|f(x_{z\Delta t})+g(x_{z\Delta t})u_{z\Delta t}\|_2$ for the unknown system. We define
\begin{equation}\label{eq:gamma}
     \gamma(u,u')=\sup_{x\in\mathcal{C}}\|g(x)\|\|u-u'\|_2,    
\end{equation}
for all $u,u'\in\mathcal{U}$, and develop the following result.
\begin{proposition}\label{prop:oracle evaluation}
Let $x_{z\Delta t}\in \mathcal{X}$, $\Theta$ be given in \eqref{eq:theta}, and $\beta$ be given in \eqref{eq:beta}. We have the following relation:
\begin{multline}
    \|f(x_{z\Delta t})+g(x_{z\Delta t})u_{z\Delta t}-f(x)-g(x)u\|_2\\
    \leq \theta(u)\|x_{z\Delta t}-x\|_2+\gamma(u_{z\Delta t},u),
\end{multline}
where $\gamma(u_{z\Delta t},u)$ is defined in \eqref{eq:gamma}.
\end{proposition}
\begin{proof}
We have that
\begin{subequations}\label{eq:oracle evaluation 1}
\begin{align}
    &\|f(x_{z\Delta t})+g(x_{z\Delta t})u_{z\Delta t}-f(x)-g(x)u\|_2\nonumber\\
    =&\|f(x_{z\Delta t})+g(x_{z\Delta t})u-f(x)-g(x)u\nonumber\\
    &\quad\quad\quad-g(x_{z\Delta t})u+g(x_{z\Delta t})u_{z\Delta t}\|_2\label{eq:oracle evaluation 1-1}\\
    \leq&\theta(u)\|x_{z\Delta t}-x\|_2+\|g(x_{z\Delta t})(u_{z\Delta t}-u_{t_k})\|_2\label{eq:oracle evaluation 1-2}\\
    \leq & \theta(u)\|x_{z\Delta t}-x\|_2+\gamma(u_{z\Delta t},u)\label{eq:oracle evaluation 1-3}
\end{align}
\end{subequations}
where \eqref{eq:oracle evaluation 1-1} holds by adding and subtracting term $g(x_{z\Delta t})u_{t_k}$, \eqref{eq:oracle evaluation 1-2} holds by triangle inequality and Proposition \ref{prop:derivative bound} in the Appendix, and \eqref{eq:oracle evaluation 1-3} holds by the fact that $\gamma(u_{z\Delta t},u_{t_k})\geq\|g(x_{z\Delta t})(u_{z\Delta t}-u_{t_k})\|_2$.
\end{proof}

Proposition \ref{prop:oracle evaluation} implies that once the value of $f(x)+g(x)u$ is known for some $x\in\mathcal{X}$ and $u\in\mathcal{U}$, we are able to calculate the range of $f(x_{z\Delta t})+g(x_{z\Delta t})u_{z\Delta t}$. In the following, we show how to construct $f(x)+g(x)u$ so as to bound $f(x_{z\Delta t})+g(x_{z\Delta t})u_{z\Delta t}$.
\begin{lemma}\label{lemma:oracle evaluation}
Let $x_{t_k},x_{t_{k+1}}\in R_K$ be two sample data points. We can construct a vector $\dot{x}$ entry-wise as
\begin{equation}\label{eq:x_t range}
    \dot{x}_{j}\doteq f_j(x) + (g(x)u_{t_k})_j=\frac{x_{t_{k+1},j}-x_{t_k,j}}{t_{k+1}-t_k}.
\end{equation}
Then the system dynamics $f(x_{z\Delta t})+g(x_{z\Delta t})u_{z\Delta t}$ satisfies
\begin{multline}\label{eq:x derivative bound}
    f(x_{z\Delta t})+g(x_{z\Delta t})u_{z\Delta t}\in
    \dot{\mathbf{x}}\\
    +\Big(\theta(u_{t_{k}})\|x_{z\Delta t}-x_{t_k}\|_2+
    \frac{\theta(u_{t_{k}})\sqrt{n}\|\beta\|_2}{\Theta}\left(e^{\Theta(t_{k+1}-t_k)}-1\right)\\
    +\gamma(u_{z\Delta t},u_{t_k})\Big)[-1,1]^n,
\end{multline}
where $\dot{\mathbf{x}}=[\dot{x},\dot{x}]$ is a thin interval.
\end{lemma}
\begin{proof}
By the mean value theorem, we have that there must exist a set of states $\{x_{\tau_j}\in\mathbb{R}^n:j=1,\ldots,n\}$ such that $\tau_j\in(t_k,t_{k+1})$ for all $j$ and the $j$-th entry of $x_{\tau_j}$, denoted as $x_{\tau_j,j}$, satisfies $x_{\tau_j,j}=\frac{x_{t_{k+1},j}-x_{t_k,j}}{t_{k+1}-t_k}$. Given the set of states $\{x_{\tau_j}\in\mathbb{R}^n:j=1,\ldots,n\}$, we can construct $\dot{x}$ entry-wise as $\dot{x}=[\dot{x}_1,\ldots,\dot{x}_n]^\top$, where $\dot{x}_j=x_{\tau_j,j}$ for each $j=1,\ldots,n$.

Let $\dot{x}$ be constructed as \eqref{eq:x_t range}. We define $\dot{\mathbf{x}}=[\dot{x},\dot{x}]$. Given $\dot{\mathbf{x}}$, we then show that \eqref{eq:x derivative bound} holds as follows.
\begin{subequations}\label{eq:x derivative bound 1}
\begin{align}
    &f(x_{z\Delta t})+g(x_{z\Delta t})u_{z\Delta t}\nonumber\\
    \in&\dot{\mathbf{x}}+(\theta(u_{t_k})\|x_{z\Delta t}-x\|_2+\gamma(u_{z\Delta t},u_{t_k}))[-1,1]^n\label{eq:x derivative bound 1-1}\\
    =&\dot{\mathbf{x}}+(\theta(u_{t_k})\|x_{z\Delta t}+x_{t_k}-x_{t_k}-x\|_2\nonumber\\
    &\quad\quad\quad\quad\quad\quad\quad+\gamma(u_{z\Delta t},u_{t_k}))[-1,1]^n\label{eq:x derivative bound 1-2}\\
    \subseteq&\dot{\mathbf{x}}+(\theta(u_{t_{k}})\|x_{z\Delta t}-x_{t_k}\|_2+\theta(u_{t_{k}})\|x_{t_k}-x\|_2\nonumber\\
    &\quad\quad\quad\quad\quad\quad\quad+\gamma(u_{z\Delta t},u_{t_k}))[-1,1]^n,\label{eq:x derivative bound 1-4}
\end{align}
\end{subequations}
where \eqref{eq:x derivative bound 1-1} holds by Proposition \ref{prop:oracle evaluation} and the fact that the sample data is generated using ZOH control input $u_t=u_{t_k}$ for all $t\in[t_k,t_{k+1})$, \eqref{eq:x derivative bound 1-2} holds by adding and subtracting term $x_{t_k}$, and \eqref{eq:x derivative bound 1-4} holds by triangle inequality. 

We prove \eqref{eq:x derivative bound} using the following relation
\begin{align*}\label{eq:x derivative bound 3}
    |x_{t_k,j}-x_{j}|&=|x_{t_k,j}-x_{\tau_j,j}|\nonumber\\
    &\leq \frac{\|\beta\|_2}{\Theta}\left(e^{\Theta(t_{k+1}-t_k)}-1\right),    
\end{align*}
where the equality holds by \eqref{eq:x_t range} and the inequality holds by Proposition \ref{prop:time bound}.
Therefore, we can bound $\|x_{t_k}-x\|_2$ as
\begin{equation}\label{eq:x derivative bound 4}
    \|x_{t_k}-x\|_2\leq \frac{\sqrt{n}\|\beta\|_2}{\Theta}\left(e^{\Theta(t_{k+1}-t_k)}-1\right).
\end{equation}
Combining \eqref{eq:x derivative bound 4} with \eqref{eq:x derivative bound 1} yields the lemma.
\end{proof}


Proposition \ref{prop:time bound} and Lemma \ref{lemma:oracle evaluation} provides us the methods to estimate the CBF constraint given in \eqref{eq:CBF with margin 1}. In the next subsection, we present an optimization problem subject to the CBF constraint for safety-critical synthesis. 

\subsection{Safety-Critical Synthesis}
In this subsection, we first use Proposition \ref{prop:time bound} and Lemma \ref{lemma:oracle evaluation} to evaluate the CBF constraint given in \eqref{eq:CBF with margin 1}. We show that the safety-critical synthesis using the evaluated CBF constraint is formulated as a non-convex program. We decompose the non-convex program to two convex sub-problems, and present an efficient safety-critical synthesis.

We define
\begin{multline*}
    w(u_{z\Delta t},u_{t_k})=\theta(u_{t_{k}})\|x_{z\Delta t}-x_{t_k}\|_2\\
    +\frac{\sqrt{n}\theta(u_{t_{k}})\|\beta\|_2}{\Theta}\left(e^{\Theta(t_{k+1}-t_k)}-1\right)+\gamma(u_{z\Delta t},u_{t_k}).
\end{multline*}
Using Proposition \ref{prop:time bound} and Lemma \ref{lemma:oracle evaluation}, we have the following upper bound for $e(x_t,x_{z\Delta t},u_{z\Delta t})$ defined in \eqref{eq:CBF error}.
\begin{theorem}\label{thm:margin bound}
Let $\dot{x}\in\mathbb{R}^n$ be constructed as \eqref{eq:x_t range}, $\mathbf{1}=[1,\ldots,1]^\top\in\mathbb{R}^n$, and
\begin{multline}
    E(\Delta t,u_{z\Delta t}) = \frac{L_h\Theta\|\beta\|_2+L_\alpha}{\Theta}(e^{\Theta\Delta t}-1)
    \\+2L_h\max\{\|\dot{x}+w(u_{z\Delta t},u_{t_k})\mathbf{1}\|_2,\|\dot{x}-w(u_{z\Delta t},u_{t_k})\mathbf{1}\|_2\}.
\end{multline}
Then $|e(x_t,x_{z\Delta t},u_{z\Delta t})|\leq E(\Delta t,u_{z\Delta t})$ for all $u_{z\Delta t}\in\mathcal{U}$.
\end{theorem}
\begin{proof}
The theorem follows from Lemma \ref{lemma:error bound}, Proposition \ref{prop:time bound}, and Lemma \ref{lemma:oracle evaluation}.
\end{proof}

Using Theorem \ref{thm:margin bound}, we have the following result.
\begin{lemma}\label{lemma:CBF with margin}
Let $\dot{x}\in\mathbb{R}^n$ be constructed as \eqref{eq:x_t range} and $\mathbf{1}=[1,\ldots,1]^\top\in\mathbb{R}^n$. If a control signal $u\in\mathcal{U}$ satisfies the following set of relations
\begin{subequations}\label{eq:modified CBF constraints}
\begin{align}
    &\frac{\partial h(x_{z\Delta t})}{\partial x}(\dot{x}+w(u_{z\Delta t},u_{t_k})\mathbf{1})+\alpha(h(x_{z\Delta t}))\nonumber\\
    &\quad\quad\quad\quad\quad\quad\quad\quad\quad\quad-E(\Delta t,u_{z\Delta t})\geq 0\\
    &\frac{\partial h(x_{z\Delta t})}{\partial x}(\dot{x}-w(u_{z\Delta t},u_{t_k})\mathbf{1})+\alpha(h(x_{z\Delta t}))\nonumber\\
    &\quad\quad\quad\quad\quad\quad\quad\quad\quad\quad-E(\Delta t,u_{z\Delta t})\geq 0
\end{align}
\end{subequations}
then control signal $u$ satisfies 
\begin{equation*}
    \frac{\partial h(x_{z\Delta t})}{\partial x}(f(x_{z\Delta t})+g(x_{z\Delta t})u)+\alpha(h(x_{z\Delta t}))\geq 0.
\end{equation*}
\end{lemma}
\begin{proof}
By Theorem \ref{thm:margin bound}, we have $|e(x_t,x_{z\Delta t},u_{z\Delta t})|\leq E(\Delta t,u_{z\Delta t})$. Therefore, if \eqref{eq:modified CBF constraints} holds, then 
\begin{subequations}\label{eq:modified CBF constraints 1}
\begin{align}
    &\frac{\partial h(x_{z\Delta t})}{\partial x}(\dot{x}+w(u,u_{t_k})\mathbf{1})+\alpha(h(x_{z\Delta t}))\nonumber\\
    &\quad\quad\quad\quad\quad\quad\quad-|e(x_t,x_{z\Delta t},u)|\geq 0\\
    &\frac{\partial h(x_{z\Delta t})}{\partial x}(\dot{x}-w(u,u_{t_k})\mathbf{1})+\alpha(h(x_{z\Delta t}))\nonumber\\
    &\quad\quad\quad\quad\quad\quad\quad-|e(x_t,x_{z\Delta t},u)|\geq 0
\end{align}
\end{subequations}
Using Lemma \ref{lemma:oracle evaluation}, we have that
\begin{multline}
    \frac{\partial h(x_{z\Delta t})}{\partial x}(f(x_{z\Delta t})+g(x_{z\Delta t})u)\\
    \in\frac{\partial h(x_{z\Delta t})}{\partial x}[\dot{x}-w(u,u_{t_k})),\dot{x}+w(u,u_{t_k}))].
\end{multline}
Hence, \eqref{eq:modified CBF constraints 1} implies that
\begin{multline*}
    \frac{\partial h(x_{z\Delta t})}{\partial x}(f(x_{z\Delta t})+g(x_{z\Delta t})u)+\alpha(h(x_{z\Delta t}))\\
    -|e(x_t,x_{z\Delta t},u)|\geq 0
\end{multline*}
Note that $|e(x_t,x_{z\Delta t},u)|\geq 0$. Therefore, the lemma holds.
\end{proof}

Motivated by Lemma \ref{lemma:CBF with margin}, we can formulate the following optimization problem at each sampling time $z\Delta t$ for all $z$
\begin{subequations}\label{eq:QP with margin}
\begin{align}
    \min_{u}~&u^\top R(x_{z\Delta t}) u\label{eq:QP with margin obj}\\
    \st ~&\frac{\partial h(x_{z\Delta t})}{\partial x}(\dot{x}+w(u,u_{t_k})\mathbf{1})+\alpha(h(x_{z\Delta t}))\nonumber\\
    &\quad\quad\quad\quad\quad\quad\quad\quad\quad\quad-E(\Delta t,u)\geq 0\label{eq:QP with margin constr 1}\\
    &\frac{\partial h(x_{z\Delta t})}{\partial x}(\dot{x}-w(u,u_{t_k})\mathbf{1})+\alpha(h(x_{z\Delta t}))\nonumber\\
    &\quad\quad\quad\quad\quad\quad\quad\quad\quad\quad-E(\Delta t,u)\geq 0\label{eq:QP with margin constr 2}\\
    &u\in\mathcal{U}\label{eq:QP with margin constr 3}
\end{align}
\end{subequations}
where $R(x_{z\Delta t})\in\mathbb{R}^m$ is a positive definite matrix. According to \eqref{eq:CBF with margin} and Lemma \ref{lemma:CBF with margin}, we have that if a control signal $u\in\mathcal{U}$ satisfying \eqref{eq:QP with margin} at sampling time $z\Delta t$ is implemented during sampling period $[z\Delta t,(z+1)\Delta t)$, then 
\begin{equation*}
    \frac{\partial h(x_t)}{\partial x}(f(x_t)+g(x_t)u)+\alpha(h(x_t))\geq 0
\end{equation*}
holds for all time $t\in [z\Delta t,(z+1)\Delta t)$. Hence, we have the following safety guarantee.
\begin{theorem}\label{thm:safety}
Let $u_{z\Delta t}^*$ be the control signal that solves \eqref{eq:QP with margin}, then system \eqref{eq:dynamic} is safe for all time $t\in[z\Delta t,(z+1)\Delta t)$ by applying control signal $u_{z\Delta t}^*$ during sampling period $t\in[z\Delta t,(z+1)\Delta t)$.
\end{theorem}
\begin{proof}
Since $u_{z\Delta t}^*$ solves optimization problem \eqref{eq:QP with margin}, constraints \eqref{eq:QP with margin constr 1} and \eqref{eq:QP with margin constr 2} hold. By Theorem \ref{thm:margin bound}, we have that $|e(x_t,x_{z\Delta t},u_{z\Delta t})|\leq E(\Delta t,u_{z\Delta t})$ for all $u_{z\Delta t}\in\mathcal{U}$. Thus, constraints \eqref{eq:QP with margin constr 1} and \eqref{eq:QP with margin constr 2} imply that
\begin{multline*}
    \frac{\partial h(x_t)}{\partial x}(f(x_t)+g(x_t)u_{z\Delta t}^*)+\alpha(h(x_t))\\
    -|e(x_t,x_{z\Delta t},u_{z\Delta t}^*)|\geq 0.
\end{multline*}
Furthermore, by the definition of $e(x_t,x_{z\Delta t},u_{z\Delta t})$ given in \eqref{eq:CBF error} and relation given in \eqref{eq:CBF with margin}, we have that constraints 
\begin{multline*}
    \frac{\partial h(x_t)}{\partial x}(f(x_t)+g(x_t)u_{z\Delta t}^*)+\alpha(h(x_t))\\
    -|e(x_t,x_{z\Delta t},u_{z\Delta t}^*)|\geq 0
\end{multline*}
holds for all $t\in[z\Delta t,(z+1)\Delta t)$. Finally, applying Lemma \ref{thm:ZCBF} yields the desired result.
\end{proof}

By Theorem \ref{thm:safety}, safety-critical synthesis reduces to solving the optimization problem given in \eqref{eq:QP with margin}. We observe that the objective function \eqref{eq:QP with margin obj} is quadratic with respect to $u$. However, the constraints \eqref{eq:QP with margin constr 1} and \eqref{eq:QP with margin constr 2} are not convex with respect to $u$. Therefore, solving problem given in \eqref{eq:QP with margin} is nontrivial. In the following, we present a two-stage approach to solve for control input $u_{z\Delta t}$ at each sampling time $z\Delta t$ for all $z=0,1,\ldots$.

We define a slack variable $p=w(u,u_{t_k})$. By Assumption \ref{assump:compact}, we have $p\in[0,\Bar{p}(u_{t_k})]$, where $\Bar{p}(u_{t_k})=\max_{u\in\mathcal{U}}\|u-u_{t_k}\|_2$. Then constraints \eqref{eq:QP with margin constr 1} and \eqref{eq:QP with margin constr 2} are rewritten as
\begin{subequations}\label{eq:QP with margin constr convex}
\begin{align}
    &2L_h\|\dot{x}+p\mathbf{1}\|_2-\frac{\partial h(x_{z\Delta t})}{\partial x}p\mathbf{1}\leq \frac{\partial h(x_{z\Delta t})}{\partial x}\dot{x}\nonumber\\
    &\quad+\alpha(h(x_{z\Delta t}))-\frac{L_h\Theta\|\beta\|_2+L_\alpha}{\Theta}(e^{\Theta\Delta t}-1)\label{eq:QP with margin constr convex 1}\\
    &2L_h\|\dot{x}-p\mathbf{1}\|_2-\frac{\partial h(x_{z\Delta t})}{\partial x}p\mathbf{1}\leq \frac{\partial h(x_{z\Delta t})}{\partial x}\dot{x}\nonumber\\
    &\quad+\alpha(h(x_{z\Delta t}))-\frac{L_h\Theta\|\beta\|_2+L_\alpha}{\Theta}(e^{\Theta\Delta t}-1)\label{eq:QP with margin constr convex 2}\\
    &2L_h\|\dot{x}+p\mathbf{1}\|_2+\frac{\partial h(x_{z\Delta t})}{\partial x}p\mathbf{1}\leq \frac{\partial h(x_{z\Delta t})}{\partial x}\dot{x}\nonumber\\
    &\quad+\alpha(h(x_{z\Delta t}))-\frac{L_h\Theta\|\beta\|_2+L_\alpha}{\Theta}(e^{\Theta\Delta t}-1)\label{eq:QP with margin constr convex 3}\\
    &2L_h\|\dot{x}-p\mathbf{1}\|_2+\frac{\partial h(x_{z\Delta t})}{\partial x}p\mathbf{1}\leq \frac{\partial h(x_{z\Delta t})}{\partial x}\dot{x}\nonumber\\
    &\quad+\alpha(h(x_{z\Delta t}))-\frac{L_h\Theta\|\beta\|_2+L_\alpha}{\Theta}(e^{\Theta\Delta t}-1)\label{eq:QP with margin constr convex 4}
\end{align}
\end{subequations}
Here $\dot{x}$ is constructed using \eqref{eq:x_t range}. Since $\|\dot{x}-p\mathbf{1}\|_2$ is convex with respect to $p$ and $\frac{\partial h(x_{z\Delta t})}{\partial x}p\mathbf{1}$ is linear with respect to $p$, the set of constraints given by \eqref{eq:QP with margin constr convex} is convex with respect to $p$. Thus, we can use the following convex program to solve for $p$ at each sampling time $z\Delta t$
\begin{subequations}\label{eq:convex program}
\begin{align}
    \min_{p}~&1\\
    \st~&p\in[0,\Bar{p}(u_{t_k})]\label{eq:convex program constr 1}\\
    &\text{constraints given by }\eqref{eq:QP with margin constr convex}\label{eq:convex program constr 2}
\end{align}
\end{subequations}
Denote the solution to convex program \eqref{eq:convex program} as $p^*$. Next, we solve for $u_{z\Delta t}$ using $p^*$. Using the definition of $w(u_{z\Delta t},u_{t_k})$ and $\gamma(u_{z\Delta t},u_{t_k})$, we define $b^*=p^*-\theta(u_{t_{k}})\|x_{z\Delta t}-x_{t_k}\|_2-\frac{\sqrt{n}\theta(u_{t_{k}})\|\beta\|_2}{\Theta}\left(e^{\Theta(t_{k+1}-t_k)}-1\right)$. Then searching for $u_{z\Delta t}$ is equivalent to computing the intersection between $\mathcal{U}$ and a ball centered at $u_{t_k}$ with radius $b^*$.
We characterize this solution procedure using the following lemma.
\begin{lemma}\label{lemma:equivalence}
If a control signal $u^*$ is feasible to \eqref{eq:QP with margin}, then $p^*=w(u^*,u_{t_k})$ solves \eqref{eq:convex program}. If there exists some $p^*$ solves \eqref{eq:convex program} and there exists some $u^*\in\mathcal{U}$ that satisfies $w(u^*,u_{t_k})=p^*$, then $u^*$ is a feasible solution to \eqref{eq:QP with margin}.
\end{lemma}
\begin{proof}
We first prove that if a control signal $u^*$ is feasible to \eqref{eq:QP with margin}, then there exists some $p^*=w(u^*,u_{t_k})$ such that $p^*$ solves \eqref{eq:convex program}. Since $u^*$ is a feasible solution to \eqref{eq:QP with margin}, constraints \eqref{eq:QP with margin constr 1} and \eqref{eq:QP with margin constr 2} hold. By the definition of $E(\Delta t,u^*)$, we observe that $p^*=w(u^*,u_{t_k})$ satisfies constraints in \eqref{eq:QP with margin constr convex}. Moreover, constraint \eqref{eq:convex program constr 1} is met since $u^*$ is feasible to \eqref{eq:QP with margin} and thus $u^*\in\mathcal{U}$. 

Next, we prove that if there exists some $p^*$ solves \eqref{eq:convex program} and there exists some $u^*\in\mathcal{U}$ that satisfies $w(u^*,u_{t_k})=p^*$, then $u^*$ is a feasible solution to \eqref{eq:QP with margin}. Using the definition of $E(\Delta t,u^*)$ and $p^*=w(u^*,u_{t_k})$, we have that constraints in \eqref{eq:QP with margin constr convex} hold implies that \eqref{eq:QP with margin constr 1} and \eqref{eq:QP with margin constr 2} hold. Additionally, we have $u^*\in\mathcal{U}$. Therefore, $u^*$ is a feasible solution to \eqref{eq:QP with margin}.

Combining the arguments above yields the lemma.
\end{proof}

By Theorem \ref{thm:safety} and Lemma \ref{lemma:equivalence}, we can compute a control signal with safety guarantee at each sampling time efficiently.

We conclude this section by discussing how the sampled-data implementation and unknown dynamics are incorporated in the proposed approach. To ensure that \eqref{eq:QP with margin} is feasible, we need 
\begin{subequations}\label{eq:sampling period analysis}
\begin{align}
    &\frac{L_h\Theta\|\beta\|_2+L_\alpha}{\Theta}(e^{\Theta\Delta t}-1)-\alpha(h(x_{z\Delta t}))\nonumber\\
    \leq &\frac{\partial h(x_{z\Delta t})}{\partial x}(\dot{x}+w(u,u_{t_k})\mathbf{1})-2L_h\nonumber\\
    &\cdot\max\{\|\dot{x}+w(u,u_{t_k})\mathbf{1}\|_2,\|\dot{x}-w(u,u_{t_k})\mathbf{1}\|_2\}\\
    &\frac{L_h\Theta\|\beta\|_2+L_\alpha}{\Theta}(e^{\Theta\Delta t}-1)-\alpha(h(x_{z\Delta t}))\nonumber\\
    \leq &\frac{\partial h(x_{z\Delta t})}{\partial x}(\dot{x}-w(u,u_{t_k})\mathbf{1})-2L_h\nonumber\\
    &\cdot\max\{\|\dot{x}+w(u,u_{t_k})\mathbf{1}\|_2,\|\dot{x}-w(u,u_{t_k})\mathbf{1}\|_2\}
\end{align}
\end{subequations}

The sampled-data implementation is captured by the term $\frac{L_h\Theta\|\beta\|_2+L_\alpha}{\Theta}(e^{\Theta\Delta t}-1)$ in \eqref{eq:sampling period analysis}. This term decreases when reducing the sampling period. When the sampling period approaches zero, we have $\frac{L_h\Theta\|\beta\|_2+L_\alpha}{\Theta}(e^{\Theta\Delta t}-1)\rightarrow 0$ since the sampled-data system approximates a continuous-time system controlled by continuous-time control signals. Additionally, reducing the sampling period $\Delta t$ is helpful when $h(x_{z\Delta t})\rightarrow 0$. The reason is that as $\Delta t\rightarrow 0$, the left-hand side of \eqref{eq:sampling period analysis} approaches zero, and thus convex program \eqref{eq:convex program} has largest feasible region.

The unknown system dynamics are captured by the terms $\dot{x}+w(u,u_{t_k})\mathbf{1}$ and $\dot{x}-w(u,u_{t_k})\mathbf{1}$ in \eqref{eq:sampling period analysis}. Reducing the sampling period will not make these terms vanish. However, we can make $w(u,u_{t_k})$ approach zero as the data set $R_K$ asymptotically covers the safe set $\mathcal{C}$. In this case, the interval derived in Lemma \ref{lemma:oracle evaluation} approaches a thin interval that only contains the system dynamics. Additionally, the data-driven methods may also help to reduce the conservativeness introduced by unknown dynamics, which is subject to our future work.

%% file: simulation.tex
\section{Numerical Case Study}\label{sec:simulation}

\begin{figure}[!htp]
	\begin{center}
		\begin{tabular}{c}
			\scalebox{0.19}{\includegraphics[width=15in]{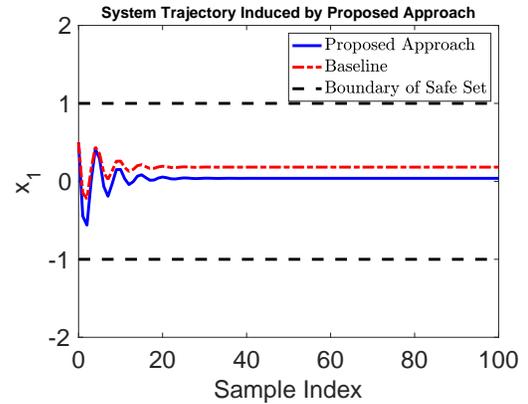}}
		\end{tabular}
		\caption{The trajectory induced by the proposed approach.}
		\label{fig:traj}
	\end{center}
\end{figure}

In this section, we present a numerical case study on the control synthesis for a DC motor. The system follows a control affine dynamics as follows \cite{daniel1998experimental}
\begin{equation*}
    \Dot{x}=\begin{bmatrix}
    -39.3153x_1 +19.1083 \\
    -1.6599x_2-3.3333
    \end{bmatrix}+\begin{bmatrix}
    -32.2293x_2\\
    22.9478x_1
    \end{bmatrix}u,
\end{equation*}
where $x_1$ is the rotor current, $x_2$ is the angular velocity, $u\in\mathcal{U}=[-4,4]$ is the stator current. The parameters in the system dynamics are unknown, and are normally obtained empirically in practice. The initial state is set as $x_0=[0.5,0.75]^\top$. The safe set of the motor is defined as $\mathcal{C}=\{x:1-x_1^2\geq 0\}$. We generate data set $R_K$ by randomly generating $200$ trajectories over $1000$ sampling periods, with the sampling period $\Delta t=0.01s$.

We compare our proposed approach with a baseline scenario. In the baseline scenario, the system dynamics is known. The control signal applied at each sampling time is calculated by solving the following quadratic program at each sampling time $z\Delta t$ for all $z=0,1,\ldots,$.
\begin{subequations}\label{eq:baseline}
\begin{align}
    \min_{u_{z\Delta t}}~&u_{z\Delta t}^2\\
    \st ~&\frac{\partial h(x)}{\partial x}f(x)+\frac{\partial h(x)}{\partial x}g(x)u_{z\Delta t}
    +\alpha(h(x))\geq 0\label{eq:baseline constr}
\end{align}
\end{subequations}

In Fig. \ref{fig:traj}, we present the trajectory generated using the proposed approach and the baseline. The trajectory generated using the proposed approach and the baseline are plotted using the blue solid line and the red dash-dotted line, respectively. We observe that both the proposed approach and the baseline guarantee safety of the system. However, the trajectories are not identical. There are two major reasons causing the trajectories to be different from each other. First, the system dynamics are not known when implementing the proposed approach, and thus the CBF constraints given in \eqref{eq:QP with margin constr convex} are used, while the baseline implements the CBF constraint given in \eqref{eq:baseline constr}. Second, the baseline aims at minimizing the energy used by the controller, i.e., solving quadratic program in \eqref{eq:baseline} at each sampling time, while the unknown system dynamics lead to a non-convex program presented in \eqref{eq:QP with margin}, which is solved by first solving the convex program in \eqref{eq:convex program} and then searching for control signal $u$ that satisfies Lemma \ref{lemma:equivalence}. We finally observe that the proposed approach provides some robustness compared with the baseline. That is, the trajectory generated using the proposed approach tends to stay further away from both boundaries of the safe set than the baseline.

%% file: Conclusion.tex
\section{Conclusion}\label{sec:conclusion}

In this paper, we studied the problem of safety-critical control synthesis for sampled-data systems with unknown dynamics. We constructed a CBF constraint to guarantee the safety during each sampling period. We evaluated the CBF constraint at each sampling time by bounding the reachable state and the unknown system dynamics. We formulated a non-convex program subject to the CBF constraint to calculate the safe control input. We decomposed the non-convex program into two sub-problems with only convex programs involved. We proved the synthesized controller guarantees the safety of the unknown sampled-data system. Our proposed solution was evaluated using a numerical case study. For future work, we will incorporate a data-driven method to improve the bound of the CBF constraint by learning the nonlinear system dynamics.

%% file: Appendix.tex
\appendix

\begin{proposition}\label{prop:derivative bound}
Suppose Assumption \ref{assump:Lipschitz} holds. Let $x,x'\in\mathcal{X}$ and $u\in\mathcal{U}$. We have 
$
    \|f(x)+g(x)u-f(x')-g(x')u\|_2\leq \theta(u)\|x-x'\|_2.
$
\end{proposition}
\begin{proof}
We analyze $\|f(x)+g(x)u-f(x')-g(x')u\|_2$ element-wise. Consider the $j$-th component of $f(x)+g(x)u$ and $f(x')+g(x')u$. We have that
\begin{subequations}\label{eq:derivative bound}
\begin{align}
    &|f_j(x)+(g(x)u)_j-f_j(x')-(g(x')u)_j|\nonumber\\
    \leq &|f_j(x)-f_j(x')| + |\sum_{s=1}^m(g_{j,s}(x)-g_{j,s}(x'))u_s|\label{eq:derivative bound 1-3}\\
    \leq & \left(L_{f_j}+\sum_{s=1}^mL_{g_{j,s}}|u_s|\right)\|x-x'\|_2,\label{eq:derivative bound 1-4}
\end{align}
\end{subequations}
where 
\eqref{eq:derivative bound 1-3} holds by triangle inequality and matrix multiplication, and \eqref{eq:derivative bound 1-4} holds by Assumption \ref{assump:Lipschitz}. Given \eqref{eq:derivative bound} holds for all $j=1,\ldots,n$, we have that the proposition holds.
\end{proof}